\documentclass[11pt,twoside,letterpaper,notitlepage]{article}

\usepackage[letterpaper,margin=1.00in]{geometry}
\usepackage{times}

\usepackage[english]{babel}

\usepackage{ifthen}

\usepackage{fancyhdr}

\usepackage{graphics}

\usepackage[sf,SF]{subfigure}

\usepackage{ifthen}

\usepackage{ifthen}

\newboolean{maybeSTOC}
\newboolean{maybeFOCS}
\newboolean{maybeElsevier}
\newboolean{maybePoster}
\newboolean{maybeSIAM}
\newboolean{maybeArticle}
\newboolean{maybeReport}
\newboolean{maybeThesis}

\ifthenelse{\isundefined{\affaddr}}        
{\setboolean{maybeSTOC}{false}}
{\setboolean{maybeSTOC}{true}}

\ifthenelse{\not \isundefined{\affiliation}
  \and \not \isundefined{\Section}}        
{\setboolean{maybeFOCS}{true}}
{\setboolean{maybeFOCS}{false}}

\ifthenelse{\not \isundefined{\tnotetext}
  \and \not \isundefined{\ead}}        
{\setboolean{maybeElsevier}{true}}
{\setboolean{maybeElsevier}{false}}

\ifthenelse{\isundefined{\conference}}
{\setboolean{maybePoster}{false}}
{\setboolean{maybePoster}{true}}

\ifthenelse{\isundefined{\sicomp}}
{\setboolean{maybeSIAM}{false}}
{\setboolean{maybeSIAM}{true}}

\ifthenelse{\isundefined{\chapter}}
{\setboolean{maybeArticle}{true}}
{\setboolean{maybeArticle}{false}}

\ifthenelse{\not \isundefined{\examen} 
  \and \not \isundefined{\disputationsdatum} 
  \and \not \isundefined{\disputationslokal}}   
  {\setboolean{maybeThesis}{true}}
  {\setboolean{maybeThesis}{false}}
           
\ifthenelse{\boolean{maybeArticle} \or \boolean{maybeThesis}}      
{\setboolean{maybeReport}{false}}
{\setboolean{maybeReport}{true}}

\usepackage[latin1]{inputenc}
\usepackage{amsmath}
\usepackage{amssymb}
\usepackage{amsfonts}
\usepackage{mathtools}

\ifthenelse
{\boolean{maybeElsevier}}
{}
{\usepackage{varioref}}

\usepackage{xspace}

\ifthenelse
{\boolean{maybeSTOC} \or \boolean{maybeSIAM}}
{}
{\usepackage{local-amsthm}}

\ifthenelse
{\boolean{maybeSTOC} 
  \or \boolean{maybeFOCS}
  \or \boolean{maybeElsevier}
 \or \boolean{maybePoster}
  \or \boolean{maybeSIAM}
  \or \boolean{maybeThesis}}
{}
{\usepackage{nada-typography}}

\usepackage{ifthen}
\usepackage{xspace}
\ifthenelse
{\boolean{maybeElsevier}}
{}
{\usepackage{varioref}}

\DeclareMathAlphabet{\mathsfsl}{OT1}{cmss}{m}{sl}

\newcommand{\eqperiod}{\enspace .}

\newcommand{\formatfunctiontoset}[1]{\mathit{#1}}

\newcommand{\introduceterm}[1]{{\emph{#1}}}

\newcommand{\ie}{i.e.,\ }

\newcommand{\wolog}{without loss of generality\xspace}

\newcommand{\Bigoh}[1]{\mathrm{O} \bigl( #1 \bigr)}
\newcommand{\bigoh}[1]{\mathrm{O} ( #1 )}

\newcommand{\Bigomega}[1]{\Omega \bigl( #1 \bigr)}
\newcommand{\bigomega}[1]{\Omega ( #1 )}

\newcommand{\problemlanguageformat}[1]{\textsc{#1}\xspace}

\newcommand{\MINIMALUNSATISFIABILITY}%
  {\problemlanguageformat{minimal unsatisfiability}}

\newcommand{\complclassformat}[1]{\textsf{#1}\xspace}
\newcommand{\cocomplclass}[1]%
        {\mbox{\complclassformat{co}-\complclassformat{#1}}\xspace}

\newcommand{\refsec}[1]{Section~\ref{#1}}

\newcommand{\refth}[1]{Theorem~\ref{#1}}

\newcommand{\reflem}[1]{Lemma~\ref{#1}}

\newcommand{\refdef}[1]{Definition~\ref{#1}}

\newcommand{\Refth}[1]{Theorem~\ref{#1}}

\ifthenelse
{\isundefined{\refeq}}
{\newcommand{\refeq}[1]{\eqref{#1}}}
{\renewcommand{\refeq}[1]{\eqref{#1}}}

\DeclareMathOperator{\Expop}{E}

\newcommand{\twincommandJN}[6]%
    {#1#2#3\vphantom{#2#5}\mspace{-2.25mu}#4.#5#6}

\newcommand{\CondExp}[2]%
    {\Expop\twincommandJN{\bigl[}{#1}{\bigl|}{\bigr}{\,#2}{\bigr]}}
\newcommand{\CONDEXP}[2]%
     {\Expop\twincommandJN{\left[}{#1}{\left|}{\right}{\,#2}{\right]}}

\newcommand{\CondProb}[3][]%
    {\Pr_{#1}\twincommandJN{\bigl[}{#2}{\bigl|}{\bigr}{\,#3}{\bigr]}}
\newcommand{\CONDPROB}[3][]%
    {\Pr_{#1}\twincommandJN{\left[}{#2}{\left|}{\right}{\,#3}{\right]}}

\newcommand{\isdistras}[2]{\ensuremath{#1} \sim \ensuremath{#2}}

\newcommand{\setcompact}[1]{{\ensuremath{\bigl\{ #1 \bigr\}}}}

\newcommand{\setdescrcompact}[3][\mid]{{\setcompact{ #2 #1 #3 }}}

\newcommand{\Setdescr}[3][|]%
     {\twincommandJN{\bigl\{}{#2}{\bigl#1}{\bigr}{\,#3}{\bigr\}}}
\newcommand{\SETDESCR}[3][|]%
     {\twincommandJN{\left\{}{#2}{\left#1}{\right}{\,#3}{\right\}}}

\newcommand{\Setdescrbrackets}[3][|]%
     {\twincommandJN{\bigl[}{#2}{\bigl#1}{\bigr}{\,#3}{\bigr]}}
\newcommand{\SETDESCRBRACKETS}[3][|]%
     {\twincommandJN{\left[}{#2}{\left#1}{\right}{\,#3}{\right]}}

\newcommand{\setsize}[1]{\lvert#1\rvert}

\newcommand{\union}{\cup}

\newcommand{\unionSP}{\, \union \, }

\newcommand{\DisjointunionInText}%
    {{\smash{\overset{\mbox{\boldmath{.}}}{\bigcup}}}\vphantom{\bigcup}}

\newcommand{\intnfirst}[1]{[{#1}]}

\newcommand{\Lor}{\bigvee}
\newcommand{\Land}{\bigwedge}

\newcommand{\xor}{\operatorname{\textsc{xor}}}

\newcommand{\olnot}[1]{\overline{#1}}
\newcommand{\stdnot}[1]{\olnot{#1}}
\newcommand{\falsenum}{0}

\newcommand{\dnfform}{DNF for\-mu\-la\xspace}

\newcommand{\xdnf}[1]{\mbox{\ensuremath{#1}-DNF}\xspace}

\newcommand{\xdnfform}[1]{\mbox{\ensuremath{#1}-}\dnfform}
\newcommand{\kdnfform}{\xdnfform{\clwidth}}

\newcommand{\xdnfset}[1]{\xdnf{#1} set\xspace}
\newcommand{\kdnfset}{\xdnfset{\clwidth}}

\newcommand{\nvar}{n}
\newcommand{\nclause}{m}
\newcommand{\clwidth}{k}

\newcommand{\randkcnfnclwrepl}[3][\clwidth]%
        {\ensuremath{\mathcal{F}^{#2, #3}_{#1}}}
\newcommand{\randkcnfnclwreplstd}%
        {\randkcnfnclwrepl{\clwidth}{\nvar}{\nclause}}

\newcommand{\israndkcnfnclwrepl}[4]%
  {\isdistras{#1}{\randkcnfnclwrepl[#2]{#3}{#4}}}

\newcommand{\randkcnfprobcl}[3]%
        {\ensuremath{\mathcal{F}^{#2}_{#1} \bigl(#3 \bigr)}}

\newcommand{\pcfor}[4][to]{for #2 := #3 #1 #4 do}
\newcommand{\pcformath}[4][to]%
    {\pcfor[#1]{\ensuremath{#2}}{\ensuremath{#3}}{\ensuremath{#4}}}

\newcommand{\pcassigncompact}[2]{#1 := #2}
\newcommand{\pcassignmathcompact}[2]%
        {\pcassigncompact{\ensuremath{#1}}{\ensuremath{#2}}}

\newcommand{\inductionformat}[1]{\textit{#1}}

\newcommand{\BASE}[1][]
        {\inductionformat
                {%
                        \ifthenelse{\equal{#1}{}}%
                                {Base case: }%
                                {Base case (#1):}%
                }%
        }

\usepackage{hyperref}

\ifthenelse
{\boolean{maybeSTOC}}
{%
\newtheorem{theorem}{Theorem}
\newtheorem{lemma}[theorem]{Lemma}
\newtheorem{proposition}[theorem]{Proposition}
\newtheorem{corollary}[theorem]{Corollary}
\newtheorem{observation}[theorem]{Observation}

\newtheorem{definition}[theorem]{Definition}

\newtheorem{conjecture}{Conjecture}
\newtheorem{openproblem}[conjecture]{Open Problem}

\newdef{remark}{Remark}
\newdef{example}{Example}

\newcounter{unnumber}

}
{\ifthenelse
  {\boolean{maybeSIAM}}
  {%
\newtheorem{observation}[theorem]{Observation}

\newtheorem{conjecture}{Conjecture}
\newtheorem{openquestion}{Open Question}

\newtheorem{remarkinner}[theorem]{Remark}
\newtheorem{exampleinner}[theorem]{Example}

\newcommand{\exampleendmarker}{\qquad$\Diamond$}
\newcommand{\remarkendmarker}{\qquad$\Diamond$}

\newenvironment{example}                        
    {\begin{exampleinner} \rm}
    {\exampleendmarker\end{exampleinner}}

\newenvironment{remark}                        
    {\begin{remarkinner} \rm}
    {\remarkendmarker\end{remarkinner}}

\newcounter{unnumber}

\newcommand{\qedhere}{\qquad}

}
  {\ifthenelse
    {\boolean{maybeArticle} \or \boolean{maybeElsevier}}
    {%
\newtheorem{standardlocalcounter}{Dummy}[section]
\newtheorem{standardglobalcounter}{Dummy}

\theoremstyle{plain}    

\newtheorem{theorem}[standardlocalcounter]{Theorem}
\newtheorem{lemma}[standardlocalcounter]{Lemma}
\newtheorem{proposition}[standardlocalcounter]{Proposition}
\newtheorem{corollary}[standardlocalcounter]{Corollary}
\newtheorem{observation}[standardlocalcounter]{Observation}

\newtheorem{conjecturelocalcounter}[standardlocalcounter]{Conjecture}
\newtheorem{conjectureglobalcounter}[standardglobalcounter]{Conjecture}
\newtheorem{conjecture}[standardglobalcounter]{Conjecture}
\newtheorem{openquestion}[standardglobalcounter]{Open Question}
\newtheorem{openproblem}[standardglobalcounter]{Open Problem}
\newtheorem{problem}{Problem}

\theoremstyle{definition}

\newtheorem{property}[standardlocalcounter]{Property}
\newtheorem{definition}[standardlocalcounter]{Definition}
\newtheorem{claim}[standardlocalcounter]{Claim}

\theoremstyle{remark}
\newtheorem{remark}[standardlocalcounter]{Remark}
\newtheorem{example}[standardlocalcounter]{Example}

\newtheoremstyle{meta}
  {3pt}
  {3pt}
  {\scshape \small }
  {}
  {\scshape \small }
  {:}
  { }
  {}

\theoremstyle{meta}
\newtheorem{meta}{Meta comment}

\newtheoremstyle{questions}
  {3pt}
  {3pt}
  {\sffamily \slshape}
  {}
  {\bfseries \sffamily \slshape}
  {:}
  { }
  {}

\theoremstyle{questions}
\newtheorem{questions}{Open questions}

}
    {%
\newtheorem{standardlocalcounter}{Dummy}[chapter]
\newtheorem{standardglobalcounter}{Dummy}

\theoremstyle{plain}    

\newtheorem{theorem}[standardlocalcounter]{Theorem}
\newtheorem{lemma}[standardlocalcounter]{Lemma}

\newtheorem{conjecture}[standardglobalcounter]{Conjecture}

\newtheorem{openproblem}[standardglobalcounter]{Open Problem}

\theoremstyle{definition}

\newtheorem{definition}[standardlocalcounter]{Definition}

\theoremstyle{remark}

\newtheoremstyle{meta}
  {3pt}
  {3pt}
  {\scshape \small }
  {}
  {\scshape \small }
  {:}
  { }
  {}

\theoremstyle{meta}

\newtheoremstyle{questions}
  {3pt}
  {3pt}
  {\sffamily \slshape}
  {}
  {\bfseries \sffamily \slshape}
  {:}
  { }
  {}

\theoremstyle{questions}

}}}

\ifthenelse{\boolean{maybeThesis}}
{}
{\usepackage{ifpdf}}

\usepackage{graphicx}  

\ifpdf         
\fi

\ifthenelse
{\boolean{maybeSTOC} \or \boolean{maybeThesis}}
{}
{
\setcounter{secnumdepth}{3}
\setcounter{tocdepth}{3}
}

\newcount\hours
\newcount\minutes
\def\SetTime{\hours=\time
\global\divide\hours by 60
\minutes=\hours
\multiply\minutes by 60
\advance\minutes by-\time
\global\multiply\minutes by-1 }
\SetTime
\def\now{\number\hours:\ifnum\minutes<10 0\fi\number\minutes}

\newcommand{\formuladots}{\cdots}

\newcommand{\proofsystemformat}[1]{\ensuremath{\mathfrak{#1}}}

\newcommand{\resknot}[1][k]{\proofsystemformat{R}({#1})}

\newcommand{\deriveswithall}%
        {\vdash_{\!\!\!{\scriptscriptstyle \forall}}} 
\newcommand{\notderiveswithall}%
        {\nvdash_{\!\!\!{\scriptscriptstyle \forall}}} 

\newcommand{\clcfgtransitioncrammed}[2]%
        {\ensuremath{#1 \!\rightsquigarrow\! #2}}

\newcommand{\fvar}{{\ensuremath{G}}}
\newcommand{\falt}{\fvar}

\newcommand{\varx}{\ensuremath{x}}
\newcommand{\vary}{\ensuremath{y}}
\newcommand{\varz}{\ensuremath{z}}

\newcommand{\clausesetformat}[1]{\ensuremath{\mathbb{#1}}}

\newcommand{\clsd}{\clausesetformat{D}}

\newcommand{\setsofvarsorlitlarge}[2]%
        {\mathit{#1}\left({#2}\right)}
\newcommand{\setsofvarsorlit}[2]%
        {\mathit{#1}({#2})}
\newcommand{\setsofvarsorlitcompact}[2]%
        {\mathit{#1}\bigl({#2}\bigr)}
\newcommand{\setsofvarsorlitsmall}[2]
        {\mathit{#1}({#2})}

\newcommand{\setsofvarsorlitsup}[3]%
        {\mathit{#1}^{#2}({#3})}
\newcommand{\setsofvarsorlitsuplarge}[3]%
        {\mathit{#1}^{#2}\left({#3}\right)}
\newcommand{\setsofvarsorlitsupcompact}[3]%
        {\mathit{#1}^{#2}\bigl({#3}\bigr)}

\newcommand{\vars}[1]{\setsofvarsorlitsmall{Vars}{#1}}

\newcommand{\derivabbrev}[2]{\bigl( #1 \vdash #2 \bigr)}
\newcommand{\derivabbrevsmall}[2]{( #1 \vdash #2 )}
\newcommand{\derivabbrevcompact}[2]{\bigl( #1 \vdash #2 \bigr)}

\newcommand{\refutabbrevsmall}[1]{\derivabbrevsmall{#1}{\falsenum}}
\newcommand{\refutabbrevcompact}[1]{\derivabbrevcompact{#1}{\falsenum}}

\newcommand{\genericmeasure}[2]{{\mathit{#1}}_{#2}}

\newcommand{\genericrefsmall}[3]%
    {{\mathit{#1}}_{#2}\refutabbrevsmall{#3}}
\newcommand{\genericrefcompact}[3]%
    {{\mathit{#1}}_{#2}\refutabbrevcompact{#3}}
\newcommand{\genericderiv}[4]%
    {{\mathit{#1}}_{#2}\derivabbrev{#3}{#4}}
\newcommand{\genericderivsmall}[4]%
    {{\mathit{#1}}_{#2}\derivabbrevsmall{#3}{#4}}
\newcommand{\genericderivcompact}[4]%
    {{\mathit{#1}}_{#2}\derivabbrevcompact{#3}{#4}}
\newcommand{\generictaut}[3]%
    {{\mathit{#1}}_{#2}\derivabbrev{}{#3}}
\newcommand{\generictautcompact}[3]%
    {{\mathit{#1}}_{#2}\derivabbrevcompact{}{#3}}
\newcommand{\generictautsmall}[3]%
    {{\mathit{#1}}_{#2}\derivabbrevsmall{}{#3}}
\newcommand{\width}[1][]{\genericmeasure{W}{#1}}

\newcommand{\formulaformat}[1]{\ensuremath{\mathit{#1}}}
\renewcommand{\formulaformat}[1]{\mathit{#1}}

\newcommand{\transitionarrow}{\rightsquigarrow}
\newcommand{\pebcfgtransition}[2]%
    {\ensuremath{#1 \transitionarrow #2}}
\newcommand{\pebcfgtransitionsqueeze}[2]%
    {#1 \! \transitionarrow \! #2}

\newcommand{\formatpebblingprice}[1]{\text{\textsl{\textsf{#1}}}}

\newcommand{\Pebblingprice}[1]%
    {\formatpebblingprice{Peb}\bigl(#1\bigr)}
\newcommand{\pebblingpricecompact}[1]
    {\formatpebblingprice{Peb}\bigl(#1\bigr)}
\newcommand{\bwpebblingprice}[1]{\formatpebblingprice{BW-Peb}(#1)}
\newcommand{\Bwpebblingprice}[1]%
    {\formatpebblingprice{BW-Peb}\bigl(#1\bigr)}
\newcommand{\bwpebblingpricecompact}[1]
    {\formatpebblingprice{BW-Peb}\bigl(#1\bigr)}

\newcommand{\bwpebpricepersistent}[1]%
    {\formatpebblingprice{BW-Peb}^{z}(#1)}
\newcommand{\Bwpebpricepersistent}[1]%
    {\formatpebblingprice{BW-Peb}^{z}\bigl(#1\bigr)}
\newcommand{\bwpebpricevisiting}[1]%
    {\formatpebblingprice{BW-Peb}^{\emptyset}(#1)}
\newcommand{\Bwpebpricevisiting}[1]%
    {\formatpebblingprice{BW-Peb}^{\emptyset}\bigl(#1\bigr)}

\newcommand{\pebpricepersistent}[1]%
    {\formatpebblingprice{Peb}^{z}(#1)}
\newcommand{\Pebpricepersistent}[1]%
    {\formatpebblingprice{Peb}^{z}\bigl(#1\bigr)}
\newcommand{\pebpricevisiting}[1]%
    {\formatpebblingprice{Peb}^{\emptyset}(#1)}
\newcommand{\Pebpricevisiting}[1]%
    {\formatpebblingprice{Peb}^{\emptyset}\bigl(#1\bigr)}

\newcommand{\bwpebblingpriceempty}[1]%
    {\formatpebblingprice{BW-Peb}^{\emptyset}(#1)}
\newcommand{\bwpebblingpriceemptycompact}[1]%
    {\formatpebblingprice{BW-Peb}^{\emptyset}\bigl(#1\bigr)}

\newcommand{\pebdeg}{\ensuremath{d}}

\newcommand{\pebaxcompact}[2]%
        [\pebdeg]{\ensuremath{\formulaformat{Ax}^{#1} \bigl(#2 \bigr)}}

\newcommand{\pqrxvar}[6]%
    {\ensuremath{\stdnot{\varx({#1})}_{#2} \lor \stdnot{\varx({#3})}_{#4} \lor %
    \sourceclausexvar[#6]{#5}}}

\newcommand{\pqr}[6]%
    {\ensuremath{\stdnot{#1}_{#2} \lor \stdnot{#3}_{#4} \lor %
    \sourceclausenodisplay[#6]{#5}}}
\newcommand{\pqrstd}{\pqr{p}{i}{q}{j}{r}{l}}
\newcommand{\pqrall}[6]%
        {\setdescrcompact
        {\pqr{#1}{#2}{#3}{#4}{#5}{#6}}{#2,#4 \in \intnfirst{\pebdeg}}}
\newcommand{\pqrallstd}%
        {\setdescrcompact{\pqrstd}{i,j \in \intnfirst{\pebdeg}}}

\newcommand{\sourceclausexvar}[2][n]%
        {\Lor_{#1 = 1}^{\pebdeg} \varx({#2})_{#1}}
\newcommand{\subsourceclausexvar}[3][n]%
        {\Lor_{#1 = {#2}}^{\pebdeg} \varx({#3})_{#1}}

\newcommand{\sourceclausexvarnodisplay}[2][n]%
        {\textstyle \Lor_{#1 = 1}^{\pebdeg} \varx({#2})_{#1}}

\newcommand{\sourceclausenodisplay}[2][n]%
        {\textstyle \Lor_{#1 = 1}^{\pebdeg} #2_{#1}}

\newcommand{\relativisation}[1]%
    {\ensuremath{\formulaformat{Rel}\bigl(#1 \bigr)}}

\usepackage{stmaryrd} 

\newcommand{\formatfunctiontosubconfiguration}[1]{\mathsf{#1}}

\newcommand{\formatfunctiontomulti}[1]{\mathcal{#1}}

\DeclareMathOperator{\dummystar}{*}
\newcommand{\pebblingcontrNT}[2][G]%
 {\ensuremath{\dummystar\!\!\formulaformat{Peb}^{#2}_{#1}}}

\newcommand{\somenodetrueclausedeg}[2]{\formulaformat{All}_{#1}^{+}({#2})}

\newcommand{\slashedstrickenletter}[1]{{\backslash\mkern-9mu #1}}
            
\newcommand{\strikethroughcommand}[1]{\slashedstrickenletter{#1}}

\newcommand{\abovevertices}[2][G]%
    {{#1}_{#2}^{\hspace{-0.2 pt}\triangledown}}
\newcommand{\aboveverticesNR}[2][G]%
    {{#1}_{\strikethroughcommand{#2}}^{\hspace{-0.3 pt}\triangledown}}
\newcommand{\belowvertices}[2][G]%
    {{#1}^{#2}_{\hspace{-0.6 pt}\vartriangle}}
\newcommand{\belowverticesNR}[2][G]%
    {{#1}^{\strikethroughcommand{#2}}_{\hspace{-0.6 pt}\vartriangle}}

\newcommand{\lpebblingpricecompact}[1]%
    {\formatpebblingprice{L-Peb}\bigl(#1\bigr)}

\newcommand{\scnot}[2]{#1 \langle #2 \rangle}
\newcommand{\scnotcompact}[2]{#1 \bigl\langle #2 \bigr\rangle}

\newcommand{\spcanonconfcompact}[1]%
        {\formatfunctiontosubconfiguration{canon}\bigl({#1}\bigr)}

\newcommand{\spprojsubsub}[4]%
    {\formatfunctiontosubconfiguration{proj}_{\scnot{#1}{#2}}(\scnot{#3}{#4})}
\newcommand{\spprojsubsubcompact}[4]%
    {\formatfunctiontosubconfiguration{proj}_{\scnot{#1}{#2}}%
    \bigl(\scnot{#3}{#4}\bigr)}
\newcommand{\spprojsubconf}[3]%
    {\formatfunctiontosubconfiguration{proj}_{\scnot{#1}{#2}}({#3})}
\newcommand{\spprojsubconfcompact}[3]%
    {\formatfunctiontosubconfiguration{proj}_{\scnot{#1}{#2}}\bigl({#3}\bigr)}
\newcommand{\spprojconfsub}[3]%
    {\formatfunctiontosubconfiguration{proj}_{#1}(\scnot{#2}{#3})}
\newcommand{\spprojconfsubcompact}[3]%
    {\formatfunctiontosubconfiguration{proj}_{#1}\bigl(\scnot{#2}{#3}\bigr)}
\newcommand{\spprojconfconf}[2]%
    {\formatfunctiontosubconfiguration{proj}_{#1}({#2})}
\newcommand{\spprojconfconfcompact}[2]%
    {\formatfunctiontosubconfiguration{proj}_{#1}\bigl({#2}\bigr)}

\newcommand{\spclossubcompact}[2]%
        {\formatfunctiontoset{cl}\bigl(\scnotcompact{#1}{#2}\bigr)}

\newcommand{\spintersubcompact}[2]%
        {\formatfunctiontoset{int}\bigl(\scnotcompact{#1}{#2}\bigr)}

\newcommand{\spcoversubcompact}[2]%
        {\formatfunctiontoset{cover}\bigl(\scnotcompact{#1}{#2}\bigr)}

\newcommand{\spcoverconfcompact}[1]%
        {\formatfunctiontoset{cover}\bigl({#1}\bigr)}

\newcommand{\spinducedblack}[1]%
    {\formatfunctiontoset{Bl} (#1)}
\newcommand{\spinducedwhite}[1]%
    {\formatfunctiontoset{Wh} (#1)}
\newcommand{\spinducedblackcompact}[1]%
    {\formatfunctiontoset{Bl} \bigl(#1 \bigr)}
\newcommand{\spinducedwhitecompact}[1]%
    {\formatfunctiontoset{Wh} \bigl(#1 \bigr)}

\newcommand{\pathclausedeg}[2][\pebdeg]%
    {\somenodetrueclausedeg[#1]{\vertexpath{#2}}}
\newcommand{\pathclauseNRdeg}[2][\pebdeg]%
    {\somenodetrueclausedeg[#1]{\vertexpathNR{#2}}}

\newcommand{\blacktruthdegexplicit}[4]%
        {\setdescrcompact
        {{\textstyle \Lor_{#2 = 1}^{#3} {#1}_{#2}}}
        {{#1} \in {#4}}}

\newcommand{\binsubtree}[1]{T^{#1}}

\newcommand{\vertexpath}[1]{{P}^{#1}}
\newcommand{\vertexpathNR}[1]{{P}_{*}^{#1}}

\newcommand{\unrelatedNP}[1]%
        {T \setminus \bigl(\binsubtree{#1} \unionSP \vertexpath{#1} \bigr)}
\newcommand{\unrelatedsmallNP}[1]%
        {T \setminus (\binsubtree{#1} \unionSP \vertexpath{#1} )}

\newcommand{\abovelevelblockerminsizecompact}%
    [2]{L_{\succeq{#1}}\bigl({#2}\bigr)}

\newcommand{\necessaryhidingvert}[2]%
{{#1}{\scriptstyle{\llfloor {#2} \rrfloor}}}

\newcommand{\Klawepropertyprefix}{Limited hiding-cardinality\xspace}

\newcommand{\klawepropacronym}{LHC property\xspace}

\newcommand{\nongenklaweprop}%
{non-generalized \Klawepropertyprefix property\xspace}
\newcommand{\nongenklawepropacronym}%
{non-generalized \klawepropacronym}
\newcommand{\nongenklawepropacronymWithParam}%
{(non-generalized) \klawepropacronym}

\newcommand{\siblingnonreachabiblitypropertynoref}%
{Sibling non-reachability property\xspace}
\newcommand{\Siblingnonreachabiblitypropertynoref}%
{Sibling non-reachability property\xspace}
\newcommand{\siblingnonreachabiblityproperty}%
{\siblingnonreachabiblitypropertynoref~%
\ref{property:sibling-non-reachability-property}\xspace}
\newcommand{\Siblingnonreachabiblityproperty}%
{\Siblingnonreachabiblitypropertynoref~%
\ref{property:sibling-non-reachability-property}\xspace}

\newcommand{\introducetermanmpctext}%
    {a \introduceterm{\mpctext{}}\xspace}

\newcommand{\introducetermamultipebblingtext}%
  {a \introduceterm{\multipebblingtext{}}\xspace}

\newcommand{\blobpebblingtext}{blob-pebbling\xspace}

\newcommand{\multipebblingtext}{\blobpebblingtext}

\newcommand{\mpcostblack}[1]%
        {\formatpebblingprice{cost}_{\mpcblacks}( #1 )}
\newcommand{\mpcostwhite}[1]%
        {\formatpebblingprice{cost}_{\mpcwhites}( #1 )}

\newcommand{\blobpebblingpricecompact}[1]%
    {\formatpebblingprice{Blob-Peb}\bigl(#1\bigr)}

\newcommand{\multipebblingpricecompact}[1]%
    {\formatpebblingprice{Blob-Peb}\bigl(#1\bigr)}

\newcommand{\mpcblacks}{\formatfunctiontomulti{B}}
\newcommand{\mpcwhites}{\formatfunctiontomulti{W}}

\newcommand{\mpscnotcompact}[2]%
        {\big[ {#1} \big] \bigl\langle {#2} \bigr\rangle}

\newcommand{\mpctext}{\blobpebblingtext con\-fig\-u\-ra\-tion\xspace}

\newcommand{\chargeablevertices}[1]%
{\formatfunctiontoset{chargeable}({#1}) }
\newcommand{\chargeableverticescompact}[1]%
{\formatfunctiontoset{chargeable}\bigl({#1}\bigr) }

\newcommand{\blackschargedfor}[1][]%
    {\mpcblacks_{#1}}
\newcommand{\whiteschargedfor}[1][]%
    {\mpcwhites_{#1}^{\hspace{-0.3 pt}\vartriangle}}

\newcommand{\whitesbelowjustblocked}%
    {\mpcwhites_{B}^{\hspace{-0.3 pt}\vartriangle}}
\newcommand{\whitesbelowhidden}%
    {\mpcwhites_{H}^{\hspace{-0.3 pt}\vartriangle}}

\newcommand{\whitestight}%
    {\mpcwhites_{T}^{\hspace{-0.3 pt}\vartriangle}}

\newcommand{\pebcontrwithfunc}[3][G]{\formulaformat{Peb}^{#2}_{#1}[{#3}]}
\newcommand{\xorpebcontrtext}{XOR-pebbling contradiction\xspace}
\newcommand{\xorpebcontr}[2][G]{\pebcontrwithfunc[#1]{#2}{\xor}}

\newcommand{\weight}[1]{\lvert{#1}\rvert}

\renewcommand{\xor}{\oplus}
\newcommand{\Xor}{\bigoplus}

\newcommand{\cdkxorform}[1][k]%
    {\mbox{$\bigl(\land \!\! \lor \!\!  \xor^{#1}\bigr)$}-block  formula\xspace}

\newcommand{\trm}{T}

\newcommand{\varu}{u}
\newcommand{\varv}{v}

\renewcommand{\varz}{z}

\newcommand{\funcw}{W}

\newcommand{\vecx}{\vec{x}}

\renewcommand{\varu}{z}
\renewcommand{\varv}{w}

\providecommand*\showkeyslabelformat[1]{%
\fbox{\normalfont\tiny#1}%
}

\begin{document}

%
%

\title{On Minimal Unsatisfiability and  \\
  Time-Space Trade-offs for $k$-DNF Resolution}

\author{%
  Jakob Nordström%
  \thanks{%
    Research supported by
    the Royal Swedish Academy of Sciences,
    the Ericsson Research Foundation,
    the Sweden-America Foundation,
    the Foundation Olle Engkvist Byggmästare, and
    the Foundation Blanceflor Boncompagni-Ludovisi, née Bildt.}  \\
  Computer Science and Artificial Intelligence Laboratory  \\
  Massachusetts Institute of Technology  \\
  Cambridge, MA 02139, USA  \\
  \texttt{jakobn@mit.edu}
  \and
  Alexander Razborov\thanks{Part of this work was done while with
Steklov Mathematical Institute, supported by the Russian Foundation
for Basic Research, and with Toyota Technological Institute at Chicago.} \\
  Department of Computer Science  \\
  University of Chicago \\
  Chicago, IL 60637, USA \\
  \texttt{razborov@cs.uchicago.edu}}

\date{\today}

\maketitle

%
%

\thispagestyle{empty}
%
%
%

%

\pagestyle{fancy}     
\fancyhead{}
\fancyfoot{}
\renewcommand{\headrulewidth}{0pt}
\renewcommand{\footrulewidth}{0pt}
%
\fancyhead[CE]{\slshape ON MINIMAL UNSATISFIABILITY
  AND TIME-SPACE TRADE-OFFS}
\fancyhead[CO]{\slshape \leftmark}
\fancyfoot[C]{\thepage}
%

\setlength{\headheight}{13.6pt}

%
%

\begin{abstract}
  In the context of proving lower bounds on proof space in $k$-DNF resolution,
  [Ben-Sasson and Nordström 2009]
  introduced the concept of minimally unsatisfiable sets of $k$-DNF
  formulas and proved that a minimally unsatisfiable $k$-DNF set with
  $m$ formulas can have at most
  $\Bigoh{(mk)^{k+1}}$ variables. They also gave an example of such sets
  with $\Omega(mk^2)$ variables.

  In this
  paper  
%
%
  we significantly improve the lower bound to $\Omega(m)^k$, which %
  almost matches the upper bound above.  Furthermore, we show
  that this implies that the analysis of their technique for proving
  time-space separations and trade-offs for $k$-DNF resolution is
  almost tight. This means that although it is possible, or even
  plausible, that stronger results than in [Ben-Sasson and Nordström
  2009] should hold, a fundamentally different approach would be
  needed to obtain such results.
\end{abstract}

\section{Introduction}
\label{sec:introduction}

A formula in conjunctive normal form, or {\em CNF formula}, is said to be {\em minimally unsatisfiable}
if it is unsatisfiable but deleting any clause makes the
formula satisfiable. A well-known result by Tarsi~%
\cite{AL86Minimal},
reproven  
several times by various authors
(see, for instance, \cite{BET01MinimallyUnsatisfiable, CS88ManyHard, Kullmann00Matroid}),
states that the number of variables in any such
CNF formula    
is always at most $(m-1)$, where $m$ is the number of clauses.

Motivated by certain problems in proof complexity related to the space
measure in the so-called {\em \mbox{$k$-DNF} resolution} proof systems
introduced by Kraj{\'\i}{\v{c}}ek~\cite{K01OnTheWeak}, Ben-Sasson and
Nordström~\cite{BN09SpaceHierachy} developed a generalization of the
concept of minimal unsatisfiability to conjunctions of formulas in
disjunctive normal form where all terms in the
disjunctions have size at most~$k$, henceforth {\em $k$-DNF
formulas}.  We begin
by  
reviewing their definition.


  Assume that $\clsd=\{D_1,\ldots,D_m\}$ is the set of $k$-DNF formulas appearing in our conjunction,
  and that $\clsd$ itself is unsatisfiable.
What should it mean that   
$\clsd$ is {\em minimally} unsatisfiable?

  The first, naive, attempt
at a definition
  would be to require, by analogy with the $k=1$ case,
  that $\clsd$ becomes satisfiable after removing
  any $D_i$ from it. However, the following simple example of two 2-DNF formulas
  \begin{equation}
  \label{eq:trivial-supposedly-minimal-formula}
\{ (x\land y_1)\lor\ldots\lor (x\land y_n),\ (\bar x_1\land y_1)\lor\ldots\lor (\bar x\land y_n)\}
  \end{equation}
  that is minimally unsatisfiable in this sense
  shows that we can not hope to get any meaningful analogue of Tarsi's lemma under
  this assumption only.

The reason for this is that the $2$-DNF set
\eqref{eq:trivial-supposedly-minimal-formula} is {\em not} minimally
unsatisfiable in the following sense: even if we ``weaken'' a formula in the
set (i.e., make it easier to satisfy) by removing any, or
even all, $y$-variables, then what remains is still an unsatisfiable
set.  This leads us to the stronger (and arguably more
natural) notion that the formula set should be minimally unsatisfiable
not only with respect to removing DNF formulas but also with respect to shrinking
terms (\ie conjunctions) in these formulas.  Fortunately, this also
turns out to be just the right notion for
%
%
the proof complexity applications given in \cite{BN09SpaceHierachy} (for details, we
refer either to that paper or to Section \ref{sec:implications-for-trade-offs} below).
Therefore, following~\cite{BN09SpaceHierachy}, we say that a set $\clsd$ of \kdnfform{}s is
\introduceterm{minimally unsatisfiable} if weakening any single term
(i.e., removing from it any literal) appearing in
a $k$-DNF formula from $\clsd$ will make the ``weaker'' set of formulas satisfiable.
This leads to the following question:
\begin{quotation}
  \noindent
  \emph{How many variables {\rm (}as a function of $k$ and $m${\rm )} may appear in a minimally
  unsatisfiable set $\{D_1,\ldots,D_m\}$ of $k$-DNF formulas?}
\end{quotation}

%
%

Tarsi's lemma thus states that for $k=1$ the answer is $(m-1)$.
This result has a relatively elementary proof based on Hall's marriage
theorem, but its importance to obtaining lower bounds on resolution
length and space is hard to overemphasize. For instance, the seminal
lower bound on refutation length of random CNF formulas in~%
\cite{CS88ManyHard} makes crucial use of it, as does the proof of the
``size-width trade-off'' in~%
\cite{BW01ShortProofs}. Examples of applications of this theorem in
resolution space lower bounds include
\cite{ABRW02SpaceComplexity,
  BG03SpaceComplexity,
  BN08ShortProofs,
  BN08UnderstandingSpace,
  NH08TowardsOptimalSeparationSTOC,
  Nordstrom09NarrowProofsSICOMP}.

To the best of our
knowledge, the case $k\geq 2$ had not been studied prior to \cite{BN09SpaceHierachy}.
That paper established  
an $\Bigoh{(mk)^{k+1}}$ upper bound and an $\Bigomega{mk^2}$ lower bound on the
number of variables. The gap is large, and, as one of their open questions, the authors
asked to narrow it.

In this
paper,  
we give an almost complete answer to that question by proving an $\Omega(m)^k$
lower bound on the number of variables. Our construction is given in  \refsec{sec:new-lower-bound},
following a little bit of preliminaries in \refsec{sec:preliminaries}. Then, in
\refsec{sec:implications-for-trade-offs}, we discuss certain consequences of our
result to proof complexity, the bottom line here being that in order to
improve on the space complexity bounds from \cite{BN09SpaceHierachy}, a
different approach
would be  
needed.
The
paper  
is concluded with a few remarks and open problems in
\refsec{sec:concluding-remarks}.

\section{Preliminaries}
\label{sec:preliminaries}

Recall that a DNF formula is a disjunction of terms, or conjunctions,
of literals, \ie unnegated or negated variables.
If all terms have size at most $k$, then the formula is
referred to as a {\em $k$-DNF formula}
(where $k$ should be thought of as some arbitrary but fixed constant).

\begin{definition}[\cite{BN09SpaceHierachy}]
  \label{def:minimal-dnf}
  A set of DNF formulas $\clsd$ is \introduceterm{minimally unsatisfiable}
  if it is unsatisfiable and
  furthermore, replacing any single term $\trm$ appearing in a single
  DNF formula $D\in\clsd$ with a proper subterm of $\trm$ makes the
  resulting set satisfiable.
\end{definition}

Note that this indeed generalizes the well-known notion of minimally
unsatisfiable CNF formulas, where a ``proper subterm'' of a literal is
the empty term $1$ that is always true and ``weakening'' a clause
hence corresponds to removing it from the formula.


We are interested in bounding the number of variables of a minimally
unsatisfiable $k$-DNF set  in terms of the number of formulas in the
set. For $1$-DNF sets (\ie CNF formulas), Tarsi's lemma \cite{AL86Minimal} states that the number of variables must be at most the number of formulas (\ie clauses) minus one for minimal unsatisfiability to hold.
This bound is easily seen to be tight by considering the
example
\begin{equation}
  \label{eq:min-unsat-CNF-formula}
  \{ x_1,x_2,\ldots,x_n,\ \bar x_1\lor \bar x_2\lor\ldots\lor \bar x_n \}
  \eqperiod
\end{equation}
No such bound holds for
general~$k$, however, since
there is an easy construction shaving off
a \mbox{factor $k^2$.}
Namely,  
denoting by $\vars{\clsd}$ the set of variables appearing somewhere in $\clsd$, we have
the following lemma.

\begin{lemma}[\cite{BN09SpaceHierachy}]
  \label{lem:explicit-construction-min-unsat-kDNF}
  There are
  arbitrarily large
  minimally unsatisfiable sets  $\clsd$ of $k$-DNF formulas
  with
  $\setsize{\vars{\clsd}}\geq k^2 ( \setsize{\clsd} - 1)$.
\end{lemma}

\begin{proof}[Proof sketch]
  Consider any minimally unsatisfiable CNF formula consisting of $n+1$
  clauses over $n$~variables (for example,
the one given in~%
  \eqref{eq:min-unsat-CNF-formula}).
  Substitute every variable $\varx_{i}$ with
  \begin{equation}
    \label{eq:explicit-construction-substitution}
    \bigl(\varx_{i}^{1} \land \varx_{i}^{2} \land \formuladots \land \varx_{i}^{k}\bigr)
    \,\lor\,
    \bigl(\varx_{i}^{k+1} \land \varx_{i}^{k+2} \land \formuladots \land \varx_{i}^{2k}\bigr)
    \,\lor\,
    \formuladots
    \,\lor\,
    \bigl(\varx_{i}^{k^2-k+1} \land \varx_{i}^{k^2-k+2} \land \formuladots \land \varx_{i}^{k^2}\bigr)
  \end{equation}
  and expand every clause to a $k$-DNF formula.
  It is straightforward to verify that the result is a minimally
  unsatisfiable \kdnfset,
  and this set has
  \mbox{$n+1$ formulas}
  over
  \mbox{$k^2 n$ variables}.
\end{proof}

There is a big gap between this lower bound on the number of variables
(in terms of the number of formulas) and the upper bound obtained in
\cite{BN09SpaceHierachy}, stated next.

\begin{theorem}[\cite{BN09SpaceHierachy}]
  \label{th:min-unsat-k-DNF-lower-bound}
  Suppose that $\clsd$ is a minimally unsatisfiable \kdnfset
  containing
  $m$
  formulas. Then
  $\setsize{\vars{\clsd}}\leq \left(km\right)^{k+1}$.
\end{theorem}

A natural problem is to close, or at least narrow, the gap between
\reflem{lem:explicit-construction-min-unsat-kDNF}
and
\refth{th:min-unsat-k-DNF-lower-bound}.
In this work, we do so by substantially improving the bound in
\reflem{lem:explicit-construction-min-unsat-kDNF}.

\section{An Improved Lower Bound for Minimally Unsatisfiable Sets}
\label{sec:new-lower-bound}

In this section, we present our construction establishing that the
number of variables in a minimally unsatisfiable $k$-DNF set can be at
least the number of formulas raised to the $k$th power.


\begin{theorem}
  \label{th:new-explicit-construction-min-unsat-kDNF}
  There exist
  arbitrarily large
  minimally unsatisfiable $k$-DNF sets $\clsd$ with
  $m$
  formulas over
  more than
  $
  \bigl(
  \frac{m}{4} \bigl( 1 - \frac{1}{k} \bigr)
  \bigr)^{k}$
  variables.
\end{theorem}

In particular, for any
$k \geq 2$
there are minimally unsatisfiable $k$-DNF sets with
$m$ formulas over (more than)
$(m/8)^k$
variables.

Very loosely,
we will use the power afforded by the $k$-terms to construct a
$k$-DNF set $\clsd$
consisting of
roughly $m$
formulas that encode
roughly $m^{k-1}$
``parallel'' instances of the minimally unsatisfiable CNF formula
in~\refeq{eq:min-unsat-CNF-formula}.
These parallel instances will be indexed by coordinate vectors
$\bigl(x^1_{i_1}, x^2_{i_2}, \ldots, x^{k-1}_{i_{k-1}}\bigr)$.
We will add auxiliary formulas enforcing that only one coordinate vector
$\bigl(x^1_{i_1}, x^2_{i_2}, \ldots, x^{k-1}_{i_{k-1}}\bigr)$
can have all coordinates true. This vector identifies which instance
of the formula~\refeq{eq:min-unsat-CNF-formula}
we are focusing on, and all other parallel instances are falsified by
their coordinate vectors not having all coordinates true.

We now formalize this loose intuition.
We first present the auxiliary formulas placing the constraints on our
coordinate vectors, which are the key to the whole construction.

\subsection{A Weight Constraint $k$-DNF Formula Set}
\label{sec:filling-in-details}

Let us write
$\vecx = \bigl( \varx_1, \ldots, \varx_{m(k-1)} \bigr)$
to denote a vector of variables of dimension $m(k-1)$.
Let
$\weight{\vecx}=\sum_{i=1}^{m(k-1)}x_i$
denote the \introduceterm{Hamming weight} of
$\vecx$,
\ie the number of ones in it.
We want to construct a $k$-DNF set
$\funcw_m(\vecx)$
with
$\bigoh{m}$
formulas over
$\varx_1, \ldots, \varx_{m(k-1)}$
and some auxiliary variables
minimally expressing that
$\weight{\vecx} \leq 1$.
That is, a vector $\vecx$
can be extended to a satisfying assignment for~%
$\funcw_m(\vecx)$ if and only if $\weight{\vecx} \leq 1$
but if we weaken any formula in the set, then there are
satisfying assignments with
$\weight{\vecx} \geq 2$.

We define
$\funcw_m(\vecx)$
to be the set of $k$-DNF formulas listed next.
The intuition for the auxiliary variables is that
$\varu_j$ can be set to true only if the first $j(k-1)$ variables
$\varx_1, \ldots, \varx_{j(k-1)}$
are all false, and
$\varv_j$ can be set to true only if at most one of
the first $j(k-1)$ variables
$\varx_1, \ldots, \varx_{j(k-1)}$
is true.
\begin{subequations}
\begin{align}
  \label{eq:weight-constraints-u1}
  &
  \olnot{\varu}_1 \lor
  \bigl( \olnot{\varx}_1 \land \formuladots \land \olnot{\varx}_{k-1}\bigr)
  \\
  \label{eq:weight-constraints-u2}
  &
  \olnot{\varu}_2 \lor
  \bigl( \varu_1 \land
  \olnot{\varx}_k \land \formuladots \land \olnot{\varx}_{2(k-1)}\bigr)
  \\
  \nonumber
  & \vdots
  \\
  \label{eq:weight-constraints-um-1}
  &
  \olnot{\varu}_{m-1} \lor
  \bigl( \varu_{m-2} \land
  \olnot{\varx}_{(m-2)(k-1)+1} \land \formuladots \land
  \olnot{\varx}_{(m-1)(k-1)}\bigr)
  \\
  \label{eq:weight-constraints-v1}
  &
  \olnot{\varv}_1 \lor \varu_1 \lor
  \Lor_{i=1}^{k}
  \Land_{\substack{i'=1 \\ i' \neq i}}^{k} \olnot{\varx}_{i'}
  \\
  \label{eq:weight-constraints-v2}
  &
  \olnot{\varv}_2 \lor \varu_2 \lor
  \bigl( \varv_1 \land
  \olnot{\varx}_k \land \formuladots \olnot{\varx}_{2(k-1)} \bigr)
  \lor
  \Lor_{i=k}^{2(k-1)}
  \biggl( \varu_{1} \land
  \Land_{\substack{i'=k \\ i' \neq i}}^{2(k-1)}
  \olnot{\varx}_{i'}
  \biggr)
  \\
  \nonumber
  & \vdots
  \\
  \label{eq:weight-constraints-vm-1}
  &
  \begin{aligned}
    \olnot{\varv}_{m-1} \lor \varu_{m-1} \lor
    \bigl( \varv_{m-2} & \land
    \olnot{\varx}_{(m-2)(k-1)+1}  \land \formuladots
    \land \olnot{\varx}_{(m-1)(k-1)} \bigr)
    \\
    &
    \lor
    \Lor_{i=(m-2)(k-1)+1}^{(m-1)(k-1)}
    \biggl( \varu_{m-2} \land
    \Land_{\substack{i'=(m-2)(k-1)+1 \\ i' \neq i}}^{(m-1)(k-1)}
    \olnot{\varx}_{i'}
    \biggr)
  \end{aligned}
  \\
  \label{eq:weight-constraints-without-vm}
  &
  \begin{aligned}
    \bigl( \varv_{m-1}  \land
    \olnot{\varx}_{(m-1)(k-1)+1}  & \land \formuladots
    \land \olnot{\varx}_{m(k-1)} \bigr)
    \\
    &
    \lor
    \Lor_{i=(m-1)(k-1)+1}^{m(k-1)}
    \biggl( \varu_{m-1} \land
    \Land_{\substack{i'=(m-1)(k-1)+1 \\ i' \neq i}}^{m(k-1)}
    \olnot{\varx}_{i'}
    \biggr).
  \end{aligned}
\end{align}
\end{subequations}
The set of $k$-DNF formulas
$\funcw_{m}$
contains
$2m-1$~formulas.
Let us see that
$\funcw_{m}$
minimally expresses that
$\vecx$ has weight at most~$1$. For
ease of notation, we will call the group of
variables $\{x_{(j-1)(k-1)+1},\ldots, x_{j(k-1)}\}$ the
\introduceterm{$j$th block} and denote it by $X_j$.

{\bf Every $\vecx$ with $\weight{\vecx}\leq 1$ can be extended to a satisfying
assignment for $\funcw_m(\vecx)$.}
Since all $x$-variables appear only negatively, we can assume
\wolog   
that $\weight{\vecx}= 1$, say
all $\varx_i$ are false except for a single variable in the $j_0$th
block $X_{j_0}$. We simply set
$\varu_j$ to true for $j < j_0$ and false for  $j \geq j_0$, and we
set all $\varv_j$ to true.

{\bf Every satisfying assignment for $\funcw_m(\vecx)$ satisfies
$\weight{\vecx}\leq 1$.}
Assume on the contrary that $\varx_{i_1}=\varx_{i_2}=1$; $i_1\in X_{j_1},\ i_2\in X_{j_2}$;
$j_1\leq j_2$.
We have that the truth of
$\varx_{i_1}$
forces
$\varu_j$
to false for all $j \geq j_1$,
and then
$\varx_{i_2}=1$
forces
$\varv_j$
to false for all $j \geq j_2$.
But this means that there is no way to satisfy the final formula
\refeq{eq:weight-constraints-without-vm}.
So for all satisfying assignments it must hold that
$\weight{\vecx} \leq 1$.

{\bf After weakening any term in $\funcw_m(\vecx)$, the resulting set can be
satisfied by an assignment giving weight at least 2 to $\vecx$.}
First we notice that weakening any of the unit terms (i.e., terms of
size one) results in removing the formula in question altogether. This
can only make it easier to satisfy the whole set than if we just
shrink a $k$-term. Hence, \wolog we can
focus on shrinking the $k$-terms.
%
%
Let us consider the formulas in
$\funcw_{m}(\vecx)$
one by one.

If we remove some literal $\olnot{x}_i$ in
\refeq{eq:weight-constraints-u1}--%
\refeq{eq:weight-constraints-um-1},
we can set
$x_i = 1$
but still have
$\varu_1 = \formuladots = \varu_{m-1} = 1$.
This will allows us to set also
$x_{m(k-1)} = 1$
in
\refeq{eq:weight-constraints-without-vm}
and still satisfy the whole set of formulas although
$\weight{\vecx} \geq 2$.

If we instead remove some~$\varu_{j}\ (j\leq m-2)$ in these formulas,
then we can set all $x_i=1$ for $x_i\in X_1\cup\ldots\cup X_j$
(that already gives us weight $\geq 2$) and $z_1=\ldots=z_j=0$,
and then we set $z_{j+1}=\ldots=z_m=1$ and $x_i=0$ for $x_i\in X_{j+1}\ldots\cup\ldots X_m$.
Note that $j\leq m-2$ implies that $z_{m-1}=1$ which takes care of
\refeq{eq:weight-constraints-without-vm}, and then \refeq{eq:weight-constraints-v1}--\refeq{eq:weight-constraints-vm-1}
are satisfied simply be setting all $w_j$ to 0. This completes
the analysis of the formulas
\refeq{eq:weight-constraints-u1}--%
\refeq{eq:weight-constraints-um-1}.

In formula
\refeq{eq:weight-constraints-v1},
if we remove some $\olnot{\varx}_{i'}$
in
$\Land_{i'=1,\ i' \neq i}^{k} \olnot{\varx}_{i'}$,
then we can set
$\varx_{i} = \varx_{i'} = \varv_{1} = 1$
and extend this to a satisfying assignment for the rest of the formulas.

For the corresponding terms
$
\varu_{j-1} \land
\Land_{i'=(j-1)(k-1)+1, \ i' \neq i}^{j(k-1)}
\olnot{\varx}_j
$
in
\refeq{eq:weight-constraints-v2}--%
\refeq{eq:weight-constraints-without-vm},
if we remove some $\olnot{\varx}_{i'}$,
we can again set
$\varx_{i} = \varx_{i'} = 1$
and set $z_1=\ldots=z_{j-1}=1$ and then $w_j=\ldots=w_{m-1}=1$
to satisfy the rest of the set,
whereas removing
$\varu_{j-1}$
would allow us to assign to 1
all $x_i\in X_1\cup\ldots\cup X_{j-1}$ and then still assign
$w_j=\ldots=w_{m-1}=1$.

For the other kind of terms
$
\varv_{j-1} \land
\olnot{\varx}_{(j-1)(k-1)+1} \land \formuladots
\land \olnot{\varx}_{j(k-1)}
$
in
\refeq{eq:weight-constraints-v2}--%
\refeq{eq:weight-constraints-without-vm},
if some $\olnot{\varx}_i$ with $x_i\in X_j$
is removed,
we can set this
$\varx_i$ to true
as well as an arbitrary $x_{i'}\in X_1\cup\ldots\cup X_{j-1}$,
whereas removing
$\varv_{j-1}$
would allow as again to set to 1 all variables
in $X_1\cup\ldots X_{j-1}$.
This proves the minimality of
$\funcw_m(\vecx)$.


\subsection{The Minimally Unsatisfiable $k$-DNF Set}
\label{sec:result}

Let us write
$\vecx^j = \bigl( \varx^j_1, \varx^j_2,\ldots, \varx^j_{m(k-1)} \bigr)$,
and let
$\funcw^j_m(\vecx^j)$
be the $k$-DNF set with $\bigoh{m}$ formulas constructed above
(over disjoint sets of variables for distinct~$j$)
minimally expressing that
$\weight{\vecx^j} \leq 1$.
With this notation, let
$\clsd^k_m$
be the $k$-DNF set consisting of the following formulas:
\begin{subequations}
\begin{align}
  \label{eq:k-dnf-weight-constraints}
  &
  \funcw^j_m(\vecx^j)
  &&
  1 \leq j < k
  \\
  \label{eq:k-dnf-positive}
  &
  \Lor_{(i_1, i_2, \ldots, i_{k-1}) \in \intnfirst{m(k-1)}^{k-1}}
  \biggl(
  \varx^1_{i_1} \land
  \varx^2_{i_2} \land
  \formuladots \land \varx^{k-1}_{i_{k-1}} \land
  \vary^{\nu}_{i_1, i_2, \ldots, i_{k-1}}
  \biggr)
  &&
  1 \leq \nu \leq m(k-1)
  \\
  \label{eq:k-dnf-negative}
  &
  \bar u_\nu \, \lor \!\!
  \Lor_{(i_1, i_2, \ldots, i_{k-1}) \in \intnfirst{m(k-1)}^{k-1}}
  \biggl(
  \varx^1_{i_1} \land
  \varx^2_{i_2} \land
  \formuladots \land \varx^{k-1}_{i_{k-1}} \land
  \olnot{\vary}^{\nu}_{i_1, i_2, \ldots, i_{k-1}}
  \biggr)
  &&
  1 \leq \nu \leq m(k-1)
  \\
  \label{eq:k-dnf-z-variables}
  &
  u_1 \lor u_2 \lor \formuladots \lor u_{m(k-1)}.
\end{align}
\end{subequations}
It is worth noting that the range of the index $\nu$ does not have
any impact on the following proof of minimal unsatisfiability, and it
was set to $m(k-1)$ only to get the best numerical results.

It is easy to verify that
$\clsd^k_m$
consists of less than
$4mk$
$k$-DNF formulas over
more than
$(m(k-1))^k = \bigl( \frac{1}{4} (4mk) \bigl( 1-\frac{1}{k} \bigr) \bigr)^k$
variables.
We claim that
$\clsd^k_m$
is minimally unsatisfiable, from which
\refth{th:new-explicit-construction-min-unsat-kDNF}
follows.

To prove the claim, let us first verify that
$\clsd^k_m$
is unsatisfiable. If the CNF formulas
$\funcw^j_m(\vecx)$
in
\refeq{eq:k-dnf-weight-constraints}
are to be satisfied for all $j < k$,
then there exists
at most one  
\mbox{$(k$$-1)$-tuple}
 $(i^*_1, i^*_2, \ldots, i^*_{k-1}) \in \intnfirst{m(k-1)}^{k-1}$
such that
$
\varx^1_{i^*_1},
\varx^2_{i^*_2},
\ldots,
\varx^{k-1}_{i^*_{k-1}}
$
are all true. This forces
$\vary^{j}_{(i^*_1, i^*_2, \ldots, i^*_{k-1})}$
to true for all $\nu$
to satisfy the formulas in
\refeq{eq:k-dnf-positive}, and then
\refeq{eq:k-dnf-negative} forces all $u_\nu$ to 0,
so that \refeq{eq:k-dnf-z-variables} is falsified.
Contradiction.

Let us now argue that
$\clsd^k_m$
is not only unsatisfiable, but \emph{minimally} unsatisfiable in the
sense of
\refdef{def:minimal-dnf}.
The proof is by case analysis over the different types of formulas
in~$\clsd^k_m$.
\begin{enumerate}
\item
  If we shrink any term in
  \refeq{eq:k-dnf-weight-constraints}---say, in
  $\funcw^1_m(\vecx^1)$, then by
  the minimality property in
  \refsec{sec:filling-in-details}
  we can set some 
  $ \varx^1_{i'_1} =   \varx^1_{i''_1} = 1  $
  for
  $i'_1 \neq i''_1$
  and then fix some
  $\varx^2_{i^*_2} = \ldots = \varx^{k-1}_{i^*_{k-1}} = 1$
  without violating the remaining clauses in
  $\funcw^1_m(\vecx^1), \ldots, \funcw^{k-1}_m(\vecx^{k-1})$.
  This allows us to satisfy the formulas in
  \refeq{eq:k-dnf-positive}
  and
  \refeq{eq:k-dnf-negative}
  by setting
  $\vary^{\nu}_{(i'_1, i^*_2 \ldots, i^*_{k-1})} = 1$
  and
  $\vary^{j}_{(i''_1, i^*_2 \ldots, i^*_{k-1})} = 0$
  for all $\nu$, respectively.
  Finally, set any $u_j$ to true to satisfy
  \refeq{eq:k-dnf-z-variables}.
  This satisfies the whole \mbox{$k$-DNF} set.

\item
  Next, suppose that we shrink some term
  $
  \varx^1_{i^*_1} \land
  \varx^2_{i^*_2} \land
  \formuladots \land
  \varx^{k-1}_{i^*_{k-1}} \land
  \vary^{\nu}_{(i^*_1, \ldots, i^*_{k-1})}
  $
  in the $\nu$th  $k$-DNF formula in
  \refeq{eq:k-dnf-positive}.
  There are two subcases:
  \begin{enumerate}
  \item
    Some $\varx$-variable is removed, say,
    the variable
    $\varx^1_{i^*_1}$.
    Set
    $\varx^1_{i^*_1} = 0$
    and
    $
    \varx^2_{i^*_2} = \ldots = \varx^{k-1}_{i^*_{k-1}} =
    \vary^{\nu}_{(i^*_1, i^*_2, \ldots, i^*_{k-1})} = 1
    $.
    This satisfies the $\nu$th formula in
    \refeq{eq:k-dnf-positive}.
    Then pick some $i'_1 \neq i^*_1$ and set
    $\varx^1_{i'_1} = 1$.
    All this can be done in a way that satisfies all clauses in~%
    \refeq{eq:k-dnf-weight-constraints}
    since the weight of every $\vecx^j$ is one.
    Set
    $u_{\nu} = 1$
    and
    $u_{\nu'} = 0$ for all $\nu' \neq \nu$
    to satisfy
    \refeq{eq:k-dnf-z-variables}
    and then
    $\vary^{\nu}_{(i'_1, i^*_2, \ldots, i^*_{k-1})} = 0$
    to satisfy
    the $\nu$th formula in
    \refeq{eq:k-dnf-negative}
    (all others are satisfied by literals $\bar u_{\nu'}$, $\nu' \neq \nu$).
    The $\nu$th formula in
    \refeq{eq:k-dnf-positive}
    was satisfied above, and for all other
    $\nu' \neq \nu$
    we set
    $\vary^{\nu}_{(i'_1, i^*_2, \ldots, i^*_{k-1})} = 1$
    to satisfy the rest of the formulas in
    \refeq{eq:k-dnf-positive}.
    This satisfies the whole \mbox{$k$-DNF} set.

  \item
    The variable
    $\vary^{\nu}_{(i^*_1, \ldots, i^*_{k-1})}$
    is eliminated.
    If so, set
    $\varx^1_{i^*_1} =
    \ldots = \varx^{k-1}_{i^*_{k-1}} = 1
    $
    to satisfy the $\nu$th formula in
    \refeq{eq:k-dnf-positive},
    $u_{\nu} = 1$
    and
    $\vary^{\nu}_{(i^*_1, \ldots, i^*_{k-1})} = 0$
    to satisfy
    \refeq{eq:k-dnf-z-variables}
    and the $\nu$th formula in
    \refeq{eq:k-dnf-negative},
    and
    $u_{\nu'} = 0$
    and
    $\vary^{\nu'}_{(i^*_1, \ldots, i^*_{k-1})} = 1$
    for all $\nu' \neq \nu$
    to satisfy the rest of the formulas in
    \refeq{eq:k-dnf-positive}
    and
    \refeq{eq:k-dnf-negative}.
    This is easily extended to an assignment  satisfying
    \refeq{eq:k-dnf-weight-constraints}
    as well.
  \end{enumerate}

\item
  For the  $\nu$th formula
  in
  \refeq{eq:k-dnf-negative}, we may assume, for the same reasons as
  in \refsec{sec:filling-in-details}, that we shrink a non-trivial
  $k$-term. Then
  we again have two subcases, treated similarly.
  \begin{enumerate}
  \item
    Some $\varx$-variable is removed, say
    $\varx^1_{i^*_1}$.
    Set
    $u_{\nu} = 1$,
    $\varx^1_{i^*_1} = 0$,
    $\varx^2_{i^*_2} = \ldots = \varx^{k-1}_{i^*_{k-1}} = 1$,
    and
    $\vary^{\nu}_{(i^*_1, i^*_2, \ldots, i^*_{k-1})} = 0$.
    This satisfies
    \refeq{eq:k-dnf-z-variables}
    and the $\nu$th formula in
    \refeq{eq:k-dnf-negative}.
    Setting
    $u_{\nu'} = 0$
    for $\nu' \neq \nu$
    takes care of the rest of
    \refeq{eq:k-dnf-negative}.
    To satisfy
    \refeq{eq:k-dnf-positive},
    we pick some $i'_1 \neq i^*_1$ and set
    $\varx^1_{i'_1} = 1$,
    and then set
    $\vary^{\nu'}_{(i'_1, i^*_2, \ldots, i^*_{k-1})} = 1$
    for all~$\nu'$.
    All this can be done in a way that satisfies the weight
    constraints in~%
    \refeq{eq:k-dnf-weight-constraints}.

  \item
    The literal
    $\olnot{\vary}^{\nu}_{(i^*_1, \ldots, i^*_{k-1})}$
    is eliminated.
    If so, set
    $\varx^1_{i^*_1} =
    \ldots = \varx^{k-1}_{i^*_{k-1}} = 1
    $
    to satisfy the $\nu$th formula in
    \refeq{eq:k-dnf-negative}
    and
    $u_{\nu} = 1$
    to satisfy
    \refeq{eq:k-dnf-z-variables}.
    Setting
    $u_{\nu'} = 0$
    for
    $\nu' \neq \nu$
    takes care of the rest of
    \refeq{eq:k-dnf-negative}.
    Now we can satisfy all of
    \refeq{eq:k-dnf-positive}
    by setting
    $\vary^{\nu}_{(i^*_1, \ldots, i^*_{k-1})} = 1$
    for all~$\nu$,
    and it is once again easy to see that the weight constraints in
    \refeq{eq:k-dnf-weight-constraints}
    are also satisfied.
  \end{enumerate}
\item \refeq{eq:k-dnf-z-variables} is removed. Set all $u_\nu$ to 0, and set
all $y^\nu_{i_1,\ldots,i_k}$ to 1, then \refeq{eq:k-dnf-weight-constraints}--\refeq{eq:k-dnf-positive}
become easy to satisfy.
\end{enumerate}

This completes the proof that
$\clsd^k_m$
is minimally unsatisfiable as claimed,
and
\refth{th:new-explicit-construction-min-unsat-kDNF}
hence follows.

\section{Implications for Time-Space Trade-offs for \lowercase{$k$}-DNF Resolution}
\label{sec:implications-for-trade-offs}


Let us start this section by a quick
review of the
relevant proof complexity context.
%
%
The \introduceterm{$k$-DNF resolution} proof systems were introduced by
Kraj{\'\i}{\v{c}}ek~\cite{K01OnTheWeak} as an intermediate step
between resolution and \mbox{depth-$2$} Frege. Roughly speaking,
the $k$th member of this family, denoted henceforth by
$\resknot$, is a system for reasoning in terms of \mbox{$k$-DNF}
formulas.  For $k=1$, the lines in the proof are hence disjunctions of
literals, and the system $\resknot[1]$ is standard resolution. At the
other extreme, $\resknot[\infty]$ is equivalent to depth-$2$ Frege.

Informally, we can think of an $\resknot$-proof as being presented on a
blackboard. The allowed derivation steps are to write on the
board a clause of the CNF formula being refuted, to deduce a new
\mbox{$k$-DNF} formula from the formulas currently on the board,
or to erase formulas from the board. The
\introduceterm{length} of an $\resknot$-proof is the total number of
formulas appearing on the board (counted with repetitions) and
the \introduceterm{(formula) space} is the maximal number of formulas
simultaneously on the board at any time during the proof.

A number of works
\cite{AB04automatizabilityOfResolutionAndRelated,
  ABE02LowerBoundsWPHPandRandom,
  Alekhnovich05LowerBoundsk-DNF3-CNF,
  JN02OptimalLowerBound,
  Razborov03PseudorandomGeneratorsHard,
  SBI04SwitchingLemma,
  Segerlind05Exponential}
have shown superpolynomial lower bounds on the length of
\mbox{$k$-DNF} refutations. It has also been established in
\cite{SBI04SwitchingLemma,Segerlind05Exponential}
that the $\resknot$-family forms a strict
hierarchy with respect to proof length.
Just as in the case for standard resolution, however,
our understanding of space
complexity in $k$-DNF resolution has remained more limited.
Esteban et al.~\cite{EGM04Complexity}
established essentially optimal space lower bounds for $\resknot$ and
also proved that the family of \emph{tree-like} $\resknot$ systems
form a strict hierarchy with respect to space. They showed that there
are formulas $F_n$ of size $n$ that can be refuted in tree-like
\mbox{$(k+1)$-DNF} resolution in constant space but require space
$\Omega(n/\log^2 n)$ to be refuted in tree-like $k$-DNF resolution. It
should be pointed out, however, that
tree-like $\resknot$ for any $k \geq 1$ is strictly weaker
than standard resolution, so the results
in \cite{EGM04Complexity} left open the question of whether
there is a strict space hierarchy for (non-tree-like) $k$-DNF
resolution or not.

Recently, the first author in joint work with Ben-Sasson
\cite{BN09SpaceHierachy}
proved that Kraj{\'\i}{\v{c}}ek's family of $\resknot$ systems do indeed form a
strict hierarchy with respect to space. However, the parameters of the
separation  were much worse than for the tree-like systems in
\cite{EGM04Complexity}, namely that the
$\resknot[k+1]$-proofs have constant space but any
$\resknot$-proof requires space
$\Bigomega{\sqrt[k+1]{n/\log n}}$.
It is not clear that there has to be a $(k+1)$st root in this
bound. No matching upper bounds are known, and indeed for the special
case of
$\resknot[2]$
versus
$\resknot[1]$
the lower bound is
$\Bigomega{n/\log n}$ by~\cite{BN08UnderstandingSpace},
\ie without a square root.
Also, combining
\cite{BN09SpaceHierachy}
with results in~\cite{BN08UnderstandingSpace}
one can derive strong length-space trade-offs for \mbox{$k$-DNF}
resolution, but again a $(k+1)$st root is lost in the analysis
compared to the corresponding results for standard resolution~%
$\resknot[1]$.

Returning now to the minimally unsatisfiable $k$-DNF sets, the reason for
studying this concept in
\cite{BN09SpaceHierachy}
was that is was an interesting special case of a more general problem
arising in their proof analysis,
and that is was hoped that better \emph{upper} bounds for this special
case would translate into improvements for the general case. Although there
appears to be no such obvious translation of \emph{lower} bounds from
the special to the general case, by using the ideas from the
previous section we can show that the analysis of the particular proof
technique employed in
\cite{BN09SpaceHierachy}
is almost tight. Thus, any further substantial  improvements of the bounds
in that paper
would have
to be obtained by other methods.

We do not go into details of the proof construction in
\cite{BN09SpaceHierachy}
here, since it is rather elaborate.
Suffice it to say that the final step of the proof boils down to studying
\mbox{$k$-DNF} sets that imply Boolean functions with a particular structure,
and proving lower bounds on the size of such DNF sets in terms of the
number of variables in these Boolean functions. Having come that far
in the construction, all that remains is a purely combinatorial
problem, and no reference to space proof complexity or \mbox{$k$-DNF}
resolution is needed.

For concreteness,
below   
we restrict our attention to the case
where the Boolean functions are exclusive or.
More general functions can be considered, and have been studied in
\cite{BN09SpaceHierachy,BN08UnderstandingSpace},
and everything that will be said below applies
to such Boolean functions with appropriate (and simple) modifications.
Hence, from now on let us focus on DNF sets that minimally imply
a particular kind of formulas that we will refer to as
\introduceterm{\cdkxorform{}s}.
A \cdkxorform is
a CNF formula in which every variable
$\varx$ is replaced by
$\Xor_{i=1}^{k} \varx_i$,
where
$\varx_1, \ldots, \varx_k$
are new variables not appearing in the original formula.
Thus, literals turn into
unnegated or negated XORs, every XOR applies to exactly one ``block'' of $k$
variables,  and no XOR mixes variables from different blocks. Let us
write this down as a formal definition.

\begin{definition}
A
\introduceterm{\cdkxorform{}}
$\falt$
is a conjunction of disjunctions of negated or unnegated exclusive
ors. The variables of $G$ are divided into disjoint blocks
$\varx_1, \ldots, \varx_k$,
$\vary_1, \ldots, \vary_k$,
$\varz_1, \ldots, \varz_k$ et cetera,
of $k$~variables each, and every XOR or negated XOR is over one full
block of variables.
\end{definition}

The key behind the lower bounds on space in
\cite{BN09SpaceHierachy}
is the result that if a $k$-DNF set $\clsd$  implies a
\cdkxorform[k+1]
$\falt$ with many variables, then $\clsd$ must also be large.

\begin{theorem}[\cite{BN09SpaceHierachy}]
  \label{th:BN09-space-main-technical}
  Let $k$ be some fixed but arbitrary
  positive integer.   
  Suppose that
  $\clsd$
  is a $k$-DNF set and that
  $\falt$
  is a
  \cdkxorform[k+1]
  such that
  $\clsd$
  implies
  $\falt$,
  and furthermore that
  $\falt$
  is minimal in the sense that if we remove a single
  XOR or negated XOR from
  $\falt$
  (thus making the formula stronger),
  it no longer holds that
  $\clsd$
  implies
  $\falt$.
  Then
  $\setsize{\vars{G}}   =  \Bigoh{\setsize{\clsd}^{k+1}}$.
\end{theorem}

Using this theorem, one can get the
$\sqrt[k+1]{n / \log n}$ space separation
mentioned above
between
$k$-DNF resolution and
$(k$$+1)$-DNF resolution.
Any improvement in the exponent in the bound in
\refth{th:BN09-space-main-technical}
would immediately translate into an improved space separation, and
would also improve the time-space trade-offs one can get when
transferring the results in
\cite{BN08UnderstandingSpace}
from resolution to $k$-DNF resolution.

Prior to the current paper, the best lower bound giving limits on what
one could hope to achieve in
\refth{th:BN09-space-main-technical}
was linear, \ie
  $
  \setsize{\vars{G}}
  =
  \bigomega{\setsize{\clsd}}
  $.
%
%
Namely, let $\falt$ be a conjunction of XORs
$
(\Xor_{i=1}^{k+1} \varx_i) \land
(\Xor_{i=1}^{k+1} \vary_i) \land
(\Xor_{i=1}^{k+1} \varz_i) \land
\formuladots
$
and let
$\clsd$
be the union of the expansions of every
$\Xor_{i=1}^{k+1} \varx_i $
as a CNF formula.
For this particular structure of
$\falt$
it is also easy to prove that
$
\setsize{\vars{G}}
=
\bigoh{\setsize{\clsd}}
$ for {\em any} choice of $\clsd$,
but it has been an open question what happens when we consider general
formulas~$\falt$.

For $k=1$,
\cite{BN08UnderstandingSpace}
proved that a linear bound
$\bigoh{\setsize{\clsd}}$
in fact holds for any
set of clauses $\clsd$ and any
\cdkxorform[2]
$\falt$,
but all attempts to extend the techniques used there to the case $k>1$
have failed. And indeed, they have failed for a good reason, since
building on the construction in
\refsec{sec:new-lower-bound}
we can show that this failure is due to the fact that the best
one can hope for in
\refth{th:BN09-space-main-technical}
is
$\setsize{\vars{G}}   =  \Bigoh{\setsize{\clsd}^{k}}$.

\begin{theorem}
  \label{th:limits-BN-technique}
  For any  $k > 1$
  there are arbirarily large
  $k$-DNF sets
  $\clsd$
  of size
  $\setsize{\clsd}=m$
  and
  \cdkxorform[k+1]{}s  $\falt$
  such that
  $\clsd$
  implies
  $\falt$,
  this implication is ``precise''
  in the sense that if we remove a single
  XOR or negated XOR from
  $\falt$
  it no longer holds that
  $\clsd$
  implies the strengthened formula,
  and
  $\setsize{\vars{G}}
  \geq (k+1) \bigl[ \frac{m}{k+2}\bigl( 1 - \frac{1}{k} \bigr) \bigr]^{k}
  \geq k \bigl( \frac{m}{4k} \bigr)^k$.
\end{theorem}

\begin{proof}
We utilize all the previous notation and start with the CNF formula
\begin{equation}
  \Land_{\nu \in \intnfirst{m(k-1)}}
  \
  \Lor_{(i_1, \ldots, i_{k-1}) \in \intnfirst{m(k-1)}^{k-1}}
  y^{\nu}_ {i_1,  \ldots, i_{k-1}}
\end{equation}
and substitute an exclusive or over variables
$y^{\nu ,r}_ {i_1,  \ldots, i_{k-1}}$,
$r = 1, \ldots, k+1$,
for every variable
$y^{\nu}_ {i_1,  \ldots, i_{k-1}}$. This results in the formula
\begin{equation}
  \falt =
  \Land_{\nu \in \intnfirst{m(k-1)}}   
  \
  \Lor_{(i_1, \ldots, i_{k-1}) \in \intnfirst{m(k-1)}^{k-1}}
  \Xor_{r=1}^{k+1}
  y^{j,r}_ {i_1,  \ldots, i_{k-1}}
\end{equation}
which will be our
\cdkxorform[k+1].
Clearly, $\falt$ contains
$  
(k+1) \cdot (m(k-1))^k$ variables.
We claim that the following easy modification of the $k$-DNF set from
Section \ref{sec:result} ``precisely'' implies $\falt$
in the sense of
\refth{th:limits-BN-technique}:

\begin{subequations}
\begin{align}
  \label{eq:limits-BN-weight-constraints}
  &
  \funcw^j_m(\vecx^j)
  &&
  1 \leq j < k
  \\
  \label{eq:limits-BN-positive}
  &
  \Lor_{(i_1, \ldots, i_{k-1}) \in \intnfirst{m}^{k-1}}
  \biggl(
  \varx^1_{i_1} \land
  \formuladots \land \varx^{k-1}_{i_{k-1}} \land
  \vary^{\nu, 1}_{i_1, \ldots, i_{k-1}}
  \biggr)
  &&
  1 \leq \nu \leq m(k-1)
  \\
  \label{eq:limits-BN-negative}
  &
  \Lor_{(i_1,  \ldots, i_{k-1}) \in \intnfirst{m}^{k-1}}
  \biggl(
  \varx^1_{i_1} \land
  \formuladots \land \varx^{k-1}_{i_{k-1}} \land
  \olnot{\vary}^{\nu,r}_{i_1, \ldots, i_{k-1}}
  \biggr)
  &&
  1 \leq \nu \leq m(k-1), \ 2 \leq r \leq k+1
\end{align}
\end{subequations}
It is straightforward to verify that
$\clsd$
consists of
less than
$m(k-1)(k+1) + 2mk \leq mk(k+2)$
\mbox{$k$-DNF} formulas.
$\clsd$ implies $\falt$ since
once we have picked which variables
$\varx^{1}_ {i^*_1},
\varx^{2}_ {i^*_2}, \ldots,
\varx^{k-1}_ {i^*_{k-1}}$
should be satisfied,
$\clsd$~will force all XOR blocks
$\Xor_{r=1}^{k+1}\vary^{\nu, r}_{i^*_1, \ldots, i^*_{k-1}}$,
$j \in \intnfirst{m(k-1)}$
to true by requiring
the variable
$\vary^{\nu, 1}_{i^*_1, \ldots, i^*_{k-1}}$
to be true and all other
variables
$\vary^{\nu, r}_{i^*_1, \ldots, i^*_{k-1}}$, $r \geq 2$,
to be false.
Finally, it is also easy to verify that
$\clsd$ implies $\falt$
``precisely'' in the sense that if a single XOR block
$\Xor_{r=1}^{k+1}\vary^{\nu, r}_{i^*_1, \ldots, i^*_{k-1}}$
is removed from $\falt$,
then we can satisfy
$\clsd$
but falsify the rest of the formula~$\falt$ (the proof is very similar
to the one given in Section \ref{sec:result}).
\Refth{th:limits-BN-technique} follows.
\end{proof}

\section{Concluding Remarks and Open Problems}
\label{sec:concluding-remarks}

We conclude this paper by discussing two remaining open problems.

Firstly, the most obvious problem still open is to close the gap
between
$\bigomega{m}^k$
and
$\Bigoh{(mk)^{k+1}}$
for the number of variables that can appear in a
minimally unsatisfiable
$k$-DNF set with $m$ formulas.
There is a strongly expressed intuition in
\cite{BN09SpaceHierachy}
that it should be possible to bring down the exponent from
$k+1$ to~$k$. Hence we have the following conjecture,
where for simplicity we fix $k$ to remove it from the asymptotic notation.

\begin{conjecture}
  Suppose that $\clsd$ is a minimally unsatisfiable \kdnfset
  for some arbitrary but fixed
  positive integer~$k$.  
  Then the
  number of variables in $\clsd$ is at most
  $
  \bigoh{ \setsize{\clsd}}^{k}
  $.
\end{conjecture}

Proving this conjecture would establish asymptotically tight bounds
for minimally unsatisfiable $k$-DNF sets (ignoring factors involving
the constant~$k$).

Secondly, we again stress that the result in
\refth{th:limits-BN-technique}
does not per se imply any restrictions (that we are aware of) on what
space separations or time-space trade-offs
are possible
for $k$-DNF resolution.
The reason for this is that our improved lower bound only rules out
\emph{a particular approach} for proving better
separations and trade-offs, but it does not say anything to the effect that the
$k$-DNF resolution proof systems are strong enough to match this lower
bound.
It would be very interesting to understand better the strength
of\mbox{ $k$-DNF} resolution in this respect.
Hence we have the following  open problem
(where we refer to
\cite{BN08UnderstandingSpace}
or
\cite{Nordstrom09PebblingSurvey}
for the relevant formal definitions).

\begin{openproblem}
  Let $\xorpebcontr{k+1}$ be the
  \xorpebcontrtext
  over some
  directed acyclic
  graph~$G$.
  Is it possible that
  $k$-DNF resolution can refute
  $\xorpebcontr{k+1}$
  in space asymptotically better than the black-white pebbling price
  $\bwpebblingprice{G}$
  of~$G$?
\end{openproblem}

We remark that for standard resolution,
\ie $1$-DNF resolution,
the answer to this question  is that
\xorpebcontrtext{}s
over two or more variables
\emph{cannot} be refuted in space less than the black-white pebbling price,
as proven in
\cite{BN08UnderstandingSpace}.
For
$k$-DNF resolution with $k>1$, however,
the best known lower bound is
$\Bigomega{\sqrt[k+1]{\bwpebblingprice{G}}}$,
as shown in
\cite{BN09SpaceHierachy}.
There is a wide gap here between the upper and lower bounds since,
as far as we are aware, there are no known
$k$-DNF resolution proofs that can do
better than space linear in the
(black)  
pebbling price (which is
achievable by standard
resolution).

\section*{Acknowledgements}

The authors would like to thank Eli Ben-Sasson for getting them to
work together on this problem and for many stimulating discussions.
Also,
the first author   
is grateful to Johan Håstad,  Nati Linial,
and \mbox{Klas Markström}
for providing thoughtful comments and advice about the problem of
minimally unsatisfiable $k$-DNF sets.

%
%

\bibliography{refArticles,refBooks,refOther}


 \bibliographystyle{alpha}   

\end{document}